\documentclass[journal,letterpaper,twocolumn,twoside]{IEEEtranTCOM}

\usepackage[utf8]{inputenc} 
\usepackage[T1]{fontenc}
\usepackage{url}
\usepackage{ifthen}
\usepackage{cite}
\usepackage{amsmath}%
\usepackage{wasysym}%
\usepackage{amssymb}

\usepackage{mdwmath}
\usepackage{blindtext}
\usepackage{eqparbox}

\usepackage[english]{babel}
\usepackage{lipsum}
\usepackage{xcolor}
\usepackage{tikz}
\usepackage{graphicx}
\usepackage{algorithm}
\usepackage{enumerate,multirow}
\usepackage{algpseudocode}

\usepackage{siunitx}

\usepackage{calc}
\newcolumntype{M}[1]{>{\centering\arraybackslash}m{#1}}
\newcolumntype{N}{@{}m{0pt}@{}}

\usepackage{amsthm}

\newtheorem{theorem}{{Theorem}}

\newtheorem{proposition}[theorem]{{Proposition}}

\newtheorem{definition}{{Definition}}

\newcommand{\cL}{{\cal L}}

\DeclareMathAlphabet{\mathbfsl}{OT1}{ppl}{b}{it} 

\newcommand{\be}[1]{\begin{equation}\label{#1}}
\newcommand{\ee}{\end{equation}}

\renewcommand{\le}{\leqslant} 
\renewcommand{\leq}{\leqslant}

\newcommand{\Pref}[1]{Pro\-po\-si\-tion\,\ref{#1}}

\newcommand{\Cref}[1]{Co\-ro\-lla\-ry\,\ref{#1}}

\IEEEoverridecommandlockouts
\begin{document}
\title{Reed-Muller Subcodes: Machine Learning-Aided Design of Efficient Soft Recursive Decoding} 
\author{
	\parbox{1.8 in}{\centering Mohammad Vahid Jamali\\
		{\small EECS Department}\\\vspace{-0.1cm}
		{\small University of Michigan}\\\vspace{-0.1cm}
		{\ttfamily\bfseries\small mvjamali@umich.edu}\\ \vspace{0.3cm}
	Hessam Mahdavifar\\
	{\small EECS Department}\\\vspace{-0.1cm}
	{\small University of Michigan}\\\vspace{-0.1cm}
	{\ttfamily\bfseries\small hessam@umich.edu}}
	\and
	\parbox{1.8 in}{\centering Xiyang Liu\\
		{\small CSE Department}\\\vspace{-0.1cm}
	{\small University of Washington}\\\vspace{-0.1cm}
		{\ttfamily\bfseries\small xiyangl@cs.washington.edu}\\ \vspace{0.3cm}
	Sewoong Oh\\
	{\small CSE Department}\\\vspace{-0.1cm}
	{\small University of Washington}\\\vspace{-0.1cm}
	{\ttfamily\bfseries\small sewoong@cs.washington.edu}}
	\and
	\parbox{2.5 in}{\centering Ashok Vardhan Makkuva\\
		{\small ECE Department}\\\vspace{-0.1cm}
		{\small University of Illinois at Urbana-Champaign}\\\vspace{-0.1cm}
		{\ttfamily\bfseries\small makkuva2@illinois.edu}\\ \vspace{0.3cm}
	     {Pramod Viswanath}\\
		{\small ECE Department}\\\vspace{-0.1cm}
		{\small University of Illinois at Urbana-Champaign}\\\vspace{-0.1cm}
		{\ttfamily\bfseries\small pramodv@illinois.edu}}
}

 \maketitle
\begin{abstract}
Reed-Muller (RM) codes
are conjectured to achieve the capacity of any binary-input memoryless symmetric (BMS) channel, and are observed to have a comparable performance to that of random codes in terms of scaling laws. On the negative side, RM codes lack efficient decoders with performance close to that of a maximum likelihood decoder for general parameters. Also, they only admit certain discrete sets of rates. In this paper, we focus on subcodes of RM codes with flexible rates that can take any code dimension from $1$ to $n$, where $n$ is the blocklength. We first extend the recursive projection-aggregation (RPA) algorithm proposed recently by Ye and Abbe for decoding RM codes. To lower the complexity of our decoding algorithm, referred to as subRPA in this paper, we investigate different ways for pruning the projections. We then derive the soft-decision based version of our algorithm, called soft-subRPA, that is shown to improve upon the performance of subRPA. Furthermore, it enables training a machine learning (ML) model to search for \textit{good} sets of projections in the sense of minimizing the decoding error rate.
Training our ML model enables achieving very close to the performance of full-projection decoding with a significantly reduced number of projections. For instance, our simulation results on a $(64,14)$ RM subcode show almost identical performance for full-projection decoding and  pruned-projection decoding with $15$ projections picked via training our ML model. This is equivalent to lowering the complexity by a factor of more than $4$ without sacrificing the decoding performance.
\end{abstract}

\section{Introduction}\label{intro}
Reed-Muller (RM) codes are among the oldest families of error-correcting codes, and their origin backs to almost seven decades ago \cite{reed1954class,muller1954application}. They have received significant renewed interest after the breakthrough invention of polar codes \cite{arikan2009channel}, given the close connection between the two classes of codes. The generator matrices for both RM and polar codes are obtained from the same square matrices, the Kronecker powers of a $2 \times 2$ matrix, though by different rules for selecting rows. In fact, such a selection of rows for polar codes is channel-specific but the RM encoder picks the rows with the largest Hamming weights. Therefore, RM codes have a universal construction. Additionally, RM codes provably achieve the Shannon capacity of binary erasure channels (BECs) at any constant rate \cite{kudekar2017reed}, and that of binary symmetric channels (BSCs) at extremal rates, i.e., at rates converging to zero or one \cite{abbe2015reed}. The long-time belief that RM codes achieve the Shannon capacity over any binary-input memoryless symmetric (BMS) channel, however, still remains an open problem \cite{abbe2020reed}. RM codes are also conjectured to have characteristics similar to those of random codes in terms of both weight enumeration \cite{kaufman2012weight} and scaling laws \cite{hassani2018almost}.

Despite their excellent performance with maximum likelihood decoders, RM codes still suffer from the lack of an efficient decoding algorithm for general parameters. Among the earlier works on decoding RM codes \cite{reed1954class,dumer2004recursive,dumer2006soft2,dumer2006soft,sakkour2005decoding,saptharishi2017efficiently,santi2018decoding}, Dumer’s recursive list decoding algorithm \cite{dumer2004recursive,dumer2006soft2,dumer2006soft} provides a trade-off between the decoding complexity and error probability. In other words, it is capable of achieving close to maximum likelihood decoding performance for large enough, e.g., exponential in blocklength, list sizes. Recently, Ye and Abbe \cite{ye2020recursive} proposed a recursive projection-aggregation (RPA) algorithm for decoding RM codes. The RPA algorithm first projects the received corrupted codeword on its cosets. It then recursively decodes the projected codes to, finally, construct the decoded codeword by properly aggregating them. Very recently, building upon the projection pruning idea in \cite{ye2020recursive}, a method for reducing the complexity of the RPA algorithm has been explored in \cite{fathollahi2020sparse}.

Besides lacking an efficient decoder in general, the structure of RM codes does not allow choosing a flexible rate. To clarify this, let $k$ and $n$ denote the code dimension and blocklength, respectively. Due to the underlying Kronecker product structure of RM codes, the code blocklength is a power of two, i.e., $n=2^m$, where $m$ is a design parameter. Additionally, RM codes posses another parameter $r$, that stands for the \textit{order} of the code, where $0\leq r\leq m$. Then, given the code blocklength $n$, one can only construct RM codes with $m+1$ possible values for the code rate, each corresponding to a given code order $r$.

This research is inspired by the aforementioned two critical issues of RM codes. More specifically, we target subcodes of RM codes,
and our primary goal is to come up with low-complexity decoders for the RM subcodes. To this end, we first extend the RPA algorithm to what we call ``subRPA'' in this paper. Similar to the RPA algorithm, subRPA starts by projecting the received corrupted codeword onto the cosets. However, since the projected codes are no longer RM codes of lower orders, their corresponding generator matrices have different ranks (i.e., different code dimensions). SubRPA applies the optimal maximum a posteriori (MAP) decoder at the bottom layer given the low dimension of the projected codes at that layer. It then aggregates the reconstructions to recursively decode the received codeword. Next, we derive the soft-decision based version of our algorithm, called ``soft-subRPA'', that improves upon the performance of subRPA.
We further investigate various ways for pruning the projections to reduce the complexity of the proposed algorithms with negligible performance loss. Enabled by our soft-subRPA algorithm, we train a machine learning (ML) model to search for \textit{good} sets of projections. We also empirically investigate encoding of RM subcodes.

In the encoding part, we observe that constructing the code generator matrix with respect to a lower complexity for our algorithms results in a superior performance compared to a higher complexity generator matrix. Also, our empirical results for pruning projections suggest a superior performance for the projection sets incurring a lower decoding complexity. This together with our observation on the encoding part unravels a two-fold gain for our proposed algorithms: a better performance for a lower complexity.
Finally, we find out that carefully training our ML model provides the possibility to find the best sets of projections that achieve very close to the performance of full-projection decoding with much smaller number of projections. 

\begin{figure}[t]\vspace{-0.075in}
	\centering
	\includegraphics[trim=0.5cm 0.2cm 0 0,width=3.6in]{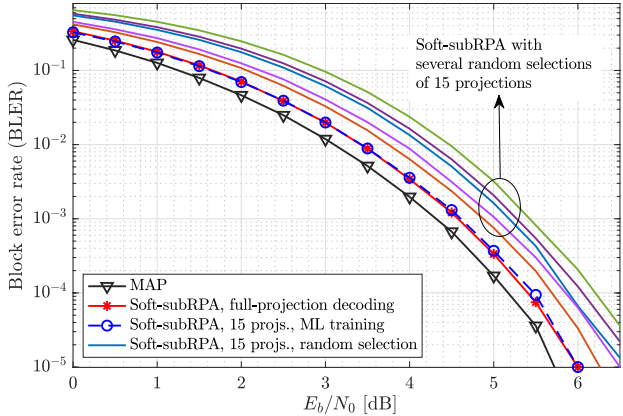}
	\caption{Performance comparison of the MAP decoder with full- and pruned-projection soft-subRPA decoding for a $(64,14)$ RM subcode.}
	\label{fig0}
	\vspace{-0.1in}
\end{figure}
Figure \ref{fig0} demonstrates the potentials of our ML-aided soft decoding algorithm, i.e., soft-subRPA with ML-aided projection pruning, in efficiently decoding RM subcodes. In this case study (detailed later in Section \ref{sec_projprun}), an RM subcode with dimension $k=14$ and blocklength $n=64$ is considered. Our ML-based projection pruning scheme, with only $15$ projections, is able to achieve an almost identical performance to that of full-projection soft-subRPA decoding with $63$ projections. This is equivalent to reducing the complexity by a factor of $4$, approximately, without sacrificing the performance. Our low-complexity ML-based pruned-projection decoding has then only about $0.25$ \si{dB} gap with the performance of the MAP decoding while randomly selecting the subsets of projections does not often provide a competitive performance.


\section{Preliminaries}\label{Prelim}
In this section, we briefly review RM codes (from an algebraic point of view) and the RPA algorithm. The reader is referred to \cite{ye2020recursive} for more details on the RPA algorithm.
\subsection{RM Codes}\label{RM_rev}
Let $k$ and $n$ denote the code dimension and blocklength, respectively. Also, let $m=\log_2 n$. The $r$-th order RM code of length $2^m$, denoted as $\mathcal{RM}(m,r)$, is then defined by the following set of vectors as the basis
\begin{align}\label{rm_basis}
\{\boldsymbol{v}_m(\mathcal{A}):~\mathcal{A}\subseteq[m],|\mathcal{A}|\le r\},
\end{align}
where $[m]:=\{1,2,\dots,m\}$, $|\mathcal{A}|$ denotes the size of the set $\mathcal{A}$, and $\boldsymbol{v}_m(\mathcal{A})$ is a row vector of length $2^m$ whose components are indexed by binary vectors $\boldsymbol{z}=(z_1,z_2,\dots,z_m) \in \{0,1\}^m$ and are defined as $\boldsymbol{v}_m(\mathcal{A},\boldsymbol{z}) = \prod_{i\in \mathcal{A}} z_i$. It can be observed from \eqref{rm_basis} that $\mathcal{RM}(m,r)$ has a dimension of $k:=\sum_{i=0}^r \binom{m}{i}$.

According to \eqref{rm_basis}, the (codebook of) $\mathcal{RM}(m,r)$ code is defined as the following set of binary vectors
\begin{align}\label{rm_codebook}
\mathcal{RM}(m,r):= \left\{\sum_{\mathcal{A}\subseteq[m],|\mathcal{A}|\le r}\hspace{-0.5cm}u(\mathcal{A}) \boldsymbol{v}_m(\mathcal{A}): u(\mathcal{A})\in\{0,1\}~
 \forall \mathcal{A}\right\}.
\end{align}
Therefore, considering a polynomial ring $\mathbb{F}_2[Z_1,Z_2,\dots,Z_m]$ of $m$ variables, the components of $\boldsymbol{v}_m(A)$ are the evaluations of the monomial $\prod_{i\in \mathcal{A}}Z_i$ at points $\boldsymbol{z}$ in the vector space $\mathbb{E}:=\mathbb{F}_2^m$. Moreover, each codeword $\boldsymbol{c}=(\boldsymbol{c}(\boldsymbol{z}), \boldsymbol{z}\in\mathbb{E})\in\mathcal{RM}(m,r)$, that is also indexed by the binary vectors $\boldsymbol{z}$, is defined as the evaluations of an $m$-variate polynomial with degree at most $r$ at points $\boldsymbol{z}\in\mathbb{E}$.

\subsection{RPA Decoding Algorithm}\label{RPA_rev}
The RPA algorithm is comprised of the following three building blocks/operations. 
\subsubsection{Projection} The RPA algorithm starts by projecting the received corrupted binary vector (in the case of BSC) or the log-likelihood ratio (LLR) vector of the channel output (in the case of general binary-input memoryless channels) into the subspaces of $\mathbb{E}$. Considering $\mathbb{B}$ as a $s$-dimensional subspace of $\mathbb{E}$, with $s\leq r$, the quotient space $\mathbb{E}/\mathbb{B}$ contains all the cosets of $\mathbb{B}$ in $\mathbb{E}$. Each coset $\boldsymbol{\tau}$ has the form $\boldsymbol{\tau}=\boldsymbol{z}+\mathbb{B}$ for some $\boldsymbol{z}\in\mathbb{E}$. Then, in the case of BSC, the projection of the channel binary output $\boldsymbol{y}=(\boldsymbol{y}(\boldsymbol{z}), \boldsymbol{z}\in\mathbb{E})$ on the cosets of $\mathbb{B}$ is defined as
\begin{align}  \label{proj_bsc}
\boldsymbol{y}_{/\mathbb{B}} :=\big(\boldsymbol{ y}_{/\mathbb{B}}(\boldsymbol{\tau}), \boldsymbol{\tau}\in\mathbb{E}/\mathbb{B} \big), \text{~s.t.~~} \boldsymbol{ y}_{/\mathbb{B}}(\boldsymbol{\tau}) := \bigoplus_{\boldsymbol{z}\in \boldsymbol{\tau}} \boldsymbol{ y}(\boldsymbol{z}),
\end{align}
where $\bigoplus$ denotes the coordinate-wise addition in $\mathbb{F}_2$. For the binary-input memoryless channels the RPA algorithm works on the projection of the channel output LLR vector $\boldsymbol{l}$. In the case of a one-dimensional subspace $\mathbb{B}$, the projected LLR vector can be obtained as $\boldsymbol{l}_{/\mathbb{B}} :=(\boldsymbol{l}_{/\mathbb{B}}(\boldsymbol{\tau}), \boldsymbol{\tau}\in\mathbb{E}/\mathbb{B})$, where
\begin{align}  \label{proj_awgn}
\boldsymbol{l}_{/\mathbb{B}}(\boldsymbol{\tau})\!=\!\ln\!\Big(\!\exp \big(\sum_{\boldsymbol{z}\in\boldsymbol{\tau}} \boldsymbol{l}(\boldsymbol{z})\big)\!+\!1\Big)\!-\!
\ln \Big( \sum_{\boldsymbol{z}\in \boldsymbol{\tau}} \exp(\boldsymbol{l}(\boldsymbol{z})) \Big).
\end{align}
\subsubsection{Decoding the Projected Outputs}
Once the decoder projects the channel output ($\boldsymbol{y}$ or $\boldsymbol{l}$), it starts recursively decoding the projected outputs, i.e., it projects them into new subspaces and continues until the projected outputs correspond to order-$1$ RM codes. The decoder then applies the fast Hadamard transform (FHT) \cite{macwilliams1977theory} to efficiently decode order-$1$ codes. Using the FHT algorithm, one can implement the MAP decoder for the first-order RM codes with complexity $\mathcal{O}(n\log n)$ instead of $\mathcal{O}(n^2)$. Once the first-order codes are decoded, the algorithm \textit{aggregates} the outputs (as explained next) to decode the codes at a higher layer. The decoder may also iterate the whole process, at each middle decoding step, several times to ensure the convergence of the algorithm.
\subsubsection{Aggregation} At each layer in the decoding process (and each point/node in the decoding tree), the decoder needs to \textit{aggregate} the output of the channel at that point with the decoding results of the next (underneath) layer to update the channel output. Note that the channel output at a given \textit{point} can be either the actual channel output ($\boldsymbol{y}$ or $\boldsymbol{l}$) or the projected ones, depending on the position of that point in the decoding tree of the recursive algorithm. Several aggregation algorithms are presented in \cite{ye2020recursive} for one- and two-dimensional subspaces. We refer the reader to \cite{ye2020recursive} for more details on the  aggregation methods.


\section{Efficient Decoding of RM Subcodes}
Let $\boldsymbol{F}=\begin{bmatrix}
1 & 0\\ 1&1
\end{bmatrix}$, and define $\boldsymbol{P}_{n\times n}=\boldsymbol{F}^{\otimes m}$, i.e., the $m$-th Kronecker power of $\boldsymbol{F}$. It can be observed that the encoding of $\mathcal{RM}(m,r)$ (described in Section \ref{RM_rev}) can be equivalently obtained by choosing the rows of the square matrix $\boldsymbol{P}_{n\times n}$ that have a Hamming weight of at least $2^{m-r}$. The resulting generator matrix $\boldsymbol{G}_{k\times n}$ then has exactly $\binom{m}{i}$ rows with the Hamming weight $n/2^i$, for $0\leq i\leq r$.

Note that the RM encoder does not allow choosing any desired code dimension; it should be of the form $k=\sum_{i=0}^{r}\binom{m}{i}$ for some $r\in\{0,1,\cdots,m\}$. Suppose that we want to construct a subcode of $\mathcal{RM}(m,r)$ with a dimension $k$ such that $k_l< k< k_u$, where $k_l :=\sum_{i=0}^{r-1}\binom{m}{i}$, $r\in[m]$, and $k_u :=\sum_{i=0}^{r}\binom{m}{i}$. Given that the construction of RM codes corresponds to picking rows of $\boldsymbol{P}_{n\times n}$ that have the highest Hamming weights, the first $k_l$ rows of the generator matrix $\boldsymbol{G}_{k\times n}$ will be the same as the generator matrix of the lower rate RM code, i.e.,  $\mathcal{RM}(m,r-1)$, that has a Hamming weight of at least $2^{m-r+1}$. It then remains to pick extra $k-k_l$ rows from $\boldsymbol{P}_{n\times n}$. These will be picked from the additional $k_u-k_l=\binom{m}{r}$ rows in $\mathcal{RM}(m,r)$ since they all have the same Hamming weight of $2^{m-r}$ which is the next largest Hamming weight. In a sense, we limit our attention to RM subcodes that, roughly speaking, \textit{sit} between two RM codes of consecutive orders. More specifically, they are subcodes of $\mathcal{RM}(m,r)$ and also contain $\mathcal{RM}(m,r-1)$ as a subcode, for some $r \in [m]$. The question is then how to choose the extra $k-k_l$ rows out of these $\binom{m}{r}$ rows of weight $2^{m-r}$ to construct an RM subcode of dimension $k$ as specified above? 
This is a very important question requiring a separate follow-up work and is beyond the scope of this paper. In the meantime, we provide some insights regarding the encoding of RM subcodes in Section \ref{sec_encoding} after describing our decoding algorithms in Sections \ref{sec_subRPA} and \ref{sec_Soft-subRPA} with respect to a generic generator matrix $\boldsymbol{G}_{k\times n}$.
Our results show that randomly selecting a subset of these rows is not always good. Indeed, some selections are better that the others, and also the set of \textit{good} rows can depend on the decoding algorithm.

\subsection{SubRPA Decoding Algorithm}\label{sec_subRPA}
Before we delve into the description of our decoding algorithms, we first need to emphasize some important facts.

\noindent{\textbf{Remark 1.}} The result of the projection operation corresponds to a code with the generator matrix that is formed by merging (i.e., binary addition of) the columns of the original code generator matrix indexed by the cosets of the projection subspace. This is clear for the BSC model, as formulated in \eqref{proj_bsc}. Additionally, for general BMS channels, the objective is to estimate the projected codewords $\boldsymbol{ c}_{/\mathbb{B}}(\boldsymbol{\tau})$'s, $\boldsymbol{\tau}\in\mathbb{E}/\mathbb{B}$, based on the channel (projected) LLRs \cite{ye2020recursive}; hence, the same principle follows for any BMS channels.


\begin{proposition}\label{prop_subcode}
Let $\mathcal{C}$ be a subcode of $\mathcal{RM}(m,r)$ with dimension $k$ such that $k_l< k< k_u$, where $k_l :=\sum_{i=0}^{r-1}\binom{m}{i}$, $r\in[m]$, and $k_u :=\sum_{i=0}^{r}\binom{m}{i}$. The projection of this code into $s$-dimensional subspaces of $\mathbb{E}$, $1\leq s\leq r-1$, results in subcodes of $\mathcal{RM}(m-s,r-s)$. It is also possible for the projected codes to be $\mathcal{RM}(m-s,r-s)$ or $\mathcal{RM}(m-s,r-1-s)$ codes.
\end{proposition}
\begin{proof}
The projection of $\mathcal{RM}(m,r)$ into $s$-dimensional subspaces, $1\leq s\leq r$ is an $\mathcal{RM}(m-s,r-s)$ code \cite{ye2020recursive}. The code $\mathcal{C}$ is constructed by removing $k_u-k$ rows of the generator matrix of $\mathcal{RM}(m,r)$ that are not in the generator matrix of $\mathcal{RM}(m,r-1)$ while the projection of $\mathcal{RM}(m,r-1)$ into $s$-dimensional subspaces, $1\leq s\leq r-1$, is an $\mathcal{RM}(m-s,r-1-s)$ code. Now, given that each $s$-dimensional projection is equivalent to partitioning $n$ columns of the generator matrix into $n/2^s$ groups of $2^s$ columns and adding them in the binary field (see Remark 1), the generator matrices of the projected codes contain rows of the generator matrix of $\mathcal{RM}(m-s,r-1-s)$ and, possibly, a subset of the rows of the generator matrix of $\mathcal{RM}(m-s,r-s)$ that are not in the generator matrix of $\mathcal{RM}(m-s,r-1-s)$. More precisely, if the selected additional $k-k_l$ rows do not contribute in the rank of the merged matrix according to a given subspace, the projected code into that subspace is an $\mathcal{RM}(m-s,r-1-s)$ code. On the other hand, if the removed $k_u-k$ rows do not contribute in that rank, the projected code is an $\mathcal{RM}(m-s,r-s)$ code. Otherwise, that projected code is a subcode of $\mathcal{RM}(m-s,r-s)$.
\end{proof}

Hereinafter, for the sake of brevity, we simply say that ``the projections of a subcode of $\mathcal{RM}(m,r)$ code into the $s$-dimensional subspaces of $\mathbb{E}$ are subcodes of $\mathcal{RM}(m-s,r-s)$''; however, we still mean the precise statement of \Pref{prop_subcode}. Now, we are ready to present our decoding algorithms for RM subcodes. Our algorithms are based on one-dimensional (1-D) subspaces. However, they can be easily generalized to the case of $s$-dimensional subspaces.


The subRPA algorithm is very similar to  the RPA algorithm. More precisely, it first projects the code $\mathcal{C}$, that is a subcode of $\mathcal{RM}(m,r)$, into 1-D subspaces to get subcodes of $\mathcal{RM}(m-1,r-1)$ at the next layer. It then recursively applies the subRPA algorithm to decode these projected codes. Next, it aggregates the decoding results of the  next layer with the output LLRs of the current layer (similar to \cite[Algorithm 4]{ye2020recursive}) to update the LLRs. Finally, it iterates this process several times to ensure the convergence of the algorithm, and takes the sign of the updated LLRs to obtain the decoded codewords.

The main distinction between subRPA algorithm and RPA algorithm, however, is the decoding of the projected codes at the bottom layer. Based on \Pref{prop_subcode}, after $r-1$ layers of 1-D projections, the decoder ends up with subcodes of $\mathcal{RM}(m-r+1,1)$ at the bottom layer. These projected codes can have different dimensions though all are less than or equal to $m-r+2$. Therefore, the subRPA algorithm, manageably, applies the MAP decoding at the bottom layer. Given that the projected codewords at the bottom layer are not all from the same codes, the MAP decoding should be carefully performed. Based on Remark 1, the projected codes at the bottom layer can be obtained from the so-called \textit{projected generator matrices} of dimension $k\times 2^{m-r+1}$, after $r-1$ times (binary) merging of the $2^m$ columns of the original generator matrix $\boldsymbol{G}_{k\times n}$. However, many of these $k$ rows of the projected generator matrices are linearly dependent. In fact, all of these matrices have ranks (i.e., code dimensions) of less than or equal to $m-r+2$. In order to facilitate the MAP decoding at the bottom layer, we can pre-compute and store the codebook of each projected code at the bottom layer. Particularly, let $R_{t}$ be the rank of the $t$-th projected generator matrix $\boldsymbol{G}_p^{(t)}$ at the bottom layer, $t\in[T]$, where $T$ is the total number of projected codes at the bottom layer (that depends on the number of layers as well as the number of projections per layer). Now, we can pre-compute the codebook $\mathcal{C}_p^{(t)}$ that contains the $2^{R_{t}}$ length-($n/2^{r-1}$) codewords $\boldsymbol{c}^{(t)}_{p,i_t}$, $i_t\in[2^{R_{t}}]$, of the $t$-th projected code at the bottom layer. Now, given the projected LLR vector $\boldsymbol{l}^{(t)}_p$ of length $n/2^{r-1}$ at the bottom layer, we pick the codeword $\boldsymbol{c}^{(t)}_{p,i^*}$ that maximizes the MAP rule for BMS channels \cite{ye2020recursive}, i.e., 
\begin{align}\label{map}
\hat{\boldsymbol{y}}_{t}=\boldsymbol{c}^{(t)}_{p,i^*},~~ \text{s.t.} ~~~~ i^*=\operatorname*{argmax}_{i_t\in[2^{R_{t}}]}~~ \langle\boldsymbol{l}^{(t)}_p,1-2{\boldsymbol{c}^{(t)}_{p,i_t}}\rangle,
\end{align}
where $\langle \cdot,\cdot\rangle$ denotes the inner (dot) product of two vectors.

\subsection{Soft-SubRPA Algorithm}\label{sec_Soft-subRPA}
In this section, we derive the soft-decision version of the subRPA algorithm, referred to as ``soft-subRPA'' in this paper. The soft-subRPA algorithm obtains soft decisions at the bottom layers instead of performing hard MAP decodings; this process is called ``soft-MAP'' in this paper. Additionally, the decoder applies a different rule to aggregate the soft decisions obtained from the next layers with the LLRs available at the current layer; we refer to this aggregation process as ``soft-aggregation''. The soft-subRPA algorithm not only improves upon the performance of the subRPA but also replaces the hard MAP decodings at the bottom layer with a differentiable operation that, in turn, enables training an ML model as delineated in Section \ref{sec_nn}.

The soft-MAP algorithm for making soft decisions on the projected codes at the bottom layer, that are subcodes of first-order RM codes, is presented in Algorithm \ref{alg_softmap} for the case of additive white Gaussian noise (AWGN) channels. The process is comprised of two main steps : 1) obtaining the LLRs of the information bits, and 2) obtaining the soft decisions (i.e., LLRs) of the coded bits using that of information bits. Note that we invoke \textit{max-log} and \textit{min-sum} approximations, to be clarified later, in Algorithm \ref{alg_softmap}. For the sake of brevity, let us drop the superscript $t$. Particularly, let $R$ be the rank of the projected generator matrix $\boldsymbol{G}_p$ of a projected code at the bottom layer with codebook $\mathcal{C}_p$. Also, assume a $2^R\times k$ matrix $\boldsymbol{U}$ that lists all $2^R$ length-$k$ sequences of bits that produce the codebook $\mathcal{C}_p$ (through modulo-$2$ matrix multiplication $\boldsymbol{U}\boldsymbol{G}_p$). Note that only $R$ indices of these length-$k$ sequences contain the information bits and the remaining indices are always fixed to either $0$ or $1$. The objective of the first step is to obtain the LLRs of these $R$ information bits using the available projected LLR vector $\boldsymbol{l}_p$. This can be done, using \eqref{llr_inf} in Appendix \ref{appnd_LLR_inf} invoking max-log approximation, as described in Algorithm \ref{alg_softmap}. Note that the LLRs of the $k-R$ indices that do not contain the information bits are set to zero.

\begin{algorithm}[t]
	\caption{Soft-MAP Algorithm for AWGN Channels}    \label{alg_softmap}
	\textbf{Input:} The LLR vector $\boldsymbol{l}_p$; the generator matrix $\boldsymbol{G}_p$; the codebook $\mathcal{C}_p$; and the matrix $\boldsymbol{U}$ of the information sequences
	
	\textbf{Output:} Soft decisions (i.e., the updated LLR vector) $\hat{\boldsymbol{l}}$
\vspace*{0.05in}
	\begin{algorithmic}[1]
		\State Set $k$ equal to the number of rows in $\boldsymbol{G}_p$
		\State  Initialize $\boldsymbol{l}_{\rm inf}$ as an all-zero vector of length $k$
		\State $\boldsymbol{\tilde{C}}\gets 1-2{\boldsymbol{C}}$ \Comment $\boldsymbol{C}$ is the codebook matrix (in binary)
		\State $\boldsymbol{\tilde{l}}\gets \boldsymbol{l}_p\boldsymbol{\tilde{C}}^{T}$ \Comment matrix mul. of $\boldsymbol{l}_p$ with the transpose of $\boldsymbol{\tilde{C}}$
		\For {$i=1,2,\cdots,k$} \Comment obtaining inf. bits LLRs
		\If {$\boldsymbol{U}(:,i)$ (the $i$-th column) is not fixed to $0$ or $1$}
		\State $\displaystyle\boldsymbol{l}_{\rm inf}(i) \gets \operatorname*{max}_{i'\in\{i':\boldsymbol{U}(i'\!,i)=0\}}\!\boldsymbol{\tilde{l}}(i')~-\hspace{-0.1cm}\operatorname*{max}_{i'\in\{i':\boldsymbol{U}(i'\!,i)=1\}}\!\boldsymbol{\tilde{l}}(i')$
		\EndIf
		\EndFor 
		\State Set $n'$ equal to the number of columns in $\boldsymbol{G}_p$
		\State  Initialize $\boldsymbol{l}_{\rm enc}$ as an all-zero vector of length $n'$
		\State $\boldsymbol{L}\gets \texttt{repeat}(\boldsymbol{l}_{\rm inf}^T,1,n')$ \Comment  make $n'$ copies of $\boldsymbol{l}_{\rm inf}^T$
		\State $\boldsymbol{V}\gets \boldsymbol{L}\odot \boldsymbol{G}_p$ \Comment element-wise matrix multiplication
		\For {$j=1,2,\cdots,n'$}
		\State $\boldsymbol{v}\gets$ vector containing nonzero elements of $\boldsymbol{V}(:,j)$
		\State $\boldsymbol{l}_{\rm enc}(j)\gets\prod_{j'}{\rm sign}(\boldsymbol{v}(j')) \times \operatorname*{min}_{j'} |\boldsymbol{v}(j')|$
		\EndFor 
		\State $\hat{\boldsymbol{l}}\gets\boldsymbol{l}_{\rm enc}$
		\State \textbf{return} $\hat{\boldsymbol{l}}$
	\end{algorithmic}
\end{algorithm}

Once we have the LLRs of the information bits, we can combine them according to the columns of $\boldsymbol{G}_p$ to obtain the LLRs of the encoded bits $\boldsymbol{l}_{\rm enc}$. The codewords in $\mathcal{C}_p$ are obtained by the multiplication of $\boldsymbol{U}\boldsymbol{G}_p$, i.e., each $j$-th coded bit, $j\in[n']$, where $n'$ is the code length, is obtained based on the linear combination of the information bits $u_i$'s according to the $j$-th column of $\boldsymbol{G}_p$. Therefore, we can apply the well-known min-sum approximation to calculate the LLR vector of the coded bits as $\boldsymbol{l}_{\rm enc}:=(\boldsymbol{l}_{\rm enc}(j), j\in[n'])$, where
\begin{align}\label{llr_enc}
\boldsymbol{l}_{\rm enc}(j)=\prod_{i\in \Delta_{j}}{\rm sign}(\boldsymbol{l}_{\rm inf}(i)) \times \operatorname*{min}_{i\in \Delta_{j}} |\boldsymbol{l}_{\rm inf}(i)|,
\end{align}
where $\Delta_{j}$ is the set of indices defining the nonzero elements in the $j$-th column of $\boldsymbol{G}_p$. This process is summarized in Algorithm \ref{alg_softmap} in an efficient way. The decoder may also iterate the whole process several times to assure the convergence of the soft-MAP algorithm.

Finally, given the soft decisions at the bottom layer, the decoder needs to aggregate the decisions with the current LLRs. 
In the following, we first define the ``soft-aggregation'' scheme as an extension of the aggregation method in \cite[Algorithm 4]{ye2020recursive} for the case of soft decisions.

\begin{definition}[Soft-Aggregation]\label{def_softaggr}
Let $\boldsymbol{l}$ be the vector of the channel LLRs, with length $n=2^m$, at a given layer. Suppose that there are $Q$ 1-D subspaces $\mathbb{B}_q$, $q\in[Q]$, to project this LLR vector at the next layer (in the case of full-projection decoding, there are $n-1$ 1-D subspaces, hence $Q=n-1$). 
Also, let $\boldsymbol{\hat{l}}_q$ denote the length-$n/2$ vector of soft decisions of the projected LLRs according to Algorithm \ref{alg_softmap}. The ``soft-aggregation'' of $\boldsymbol{l}$ and $\boldsymbol{\hat{l}}_q$'s is defined as a length-$n$ vector $\tilde{\boldsymbol{l}}:=(\tilde{\boldsymbol{l}}(\boldsymbol{z}), \boldsymbol{z}\in\mathbb{F}_2^m)$ where 
\begin{align}\label{eq_softaggr}
\tilde{\boldsymbol{l}}(\boldsymbol{z})=\frac{1}{Q}\sum_{q=1}^{Q}\tanh\big(\boldsymbol{\hat{l}}_q\left([\boldsymbol{z}+\mathbb{B}_q]\right)/2\big) \boldsymbol{l}(\boldsymbol{z}\oplus\boldsymbol{z}_q).
\end{align}
where $\boldsymbol{z}_q$ is the nonzero vector of the 1-D subspace $\mathbb{B}_q$, and $[\boldsymbol{z}+\mathbb{B}_q]$ is the coset containing $\boldsymbol{z}$ for the projection into $\mathbb{B}_q$. 
\end{definition}
In order to observe \eqref{eq_softaggr}, recall that the objective of the aggregation step is to update the length-$n$ channel LLR vector $\boldsymbol{l}$ to $\tilde{\boldsymbol{l}}$ given the soft decisions of the projected codes.  $\boldsymbol{\hat{l}}_q\left([\boldsymbol{z}+\mathbb{B}_q]\right)$ severs as a soft estimate of the binary addition of the coded bits at positions $\boldsymbol{z}$ and $\boldsymbol{z}\oplus\boldsymbol{z}_q$. Hence, following the same arguments as \cite{ye2020recursive}, if that combined bit is $0$, then the updated LLR at position $\boldsymbol{z}$ should take the same sign as the channel LLR at position $\boldsymbol{z}\oplus \boldsymbol{z}_q$. Note that this happens with probability $a_0:=1/\big[1+\exp\big(-\boldsymbol{\hat{l}}_q\left([\boldsymbol{z}+\mathbb{B}_q]\right)\big)\big]$. Similarly, with probability $a_1:=1/\big[1+\exp\big(\boldsymbol{\hat{l}}_q\left([\boldsymbol{z}+\mathbb{B}_q]\right)\big)\big]$ the combined bit is $1$, and hence the updated LLR at position $\boldsymbol{z}$ and $\boldsymbol{l}(\boldsymbol{z}\oplus\boldsymbol{z}_q)$ should have different signs. Therefore, given a projection subspace $\mathbb{B}_q$, one can update the channel LLR as $a_0\times\boldsymbol{l}(\boldsymbol{z}\oplus\boldsymbol{z}_q)+a_1\times-\boldsymbol{l}(\boldsymbol{z}\oplus\boldsymbol{z}_q)$. Taking the average over all $Q$ projections then yields the soft-aggregation rule in \eqref{eq_softaggr}.

It is worth mentioning that one can also apply the following equation to update the channel LLR as
\begin{align}\label{eq_logsum}
\tilde{\boldsymbol{l}}_{\rm ls}(\boldsymbol{z})=\frac{1}{Q}\sum_{q=1}^{Q}\ln\left(\frac{1+{\rm e}^{\boldsymbol{\hat{l}}_q\left([\boldsymbol{z}+\mathbb{B}_q]\right)+\boldsymbol{l}(\boldsymbol{z}\oplus\boldsymbol{z}_q)}}{{\rm e}^{\boldsymbol{\hat{l}}_q\left([\boldsymbol{z}+\mathbb{B}_q]\right)}+{\rm e}^{\boldsymbol{l}(\boldsymbol{z}\oplus\boldsymbol{z}_q)}}\right).
\end{align}
The rationale behind \eqref{eq_logsum} follows similar arguments as above and then deriving the LLR of the sum of two binary random variables given the LLRs of each of them. Therefore, \eqref{eq_logsum} is an exact expression assuming  independence among the involved LLR components.   Our empirical observations, however, suggest almost identical results for either aggregation methods. Therefore, given the complexity of computing expressions like \eqref{eq_logsum}, one can reliably apply our proposed soft-aggregation method in Definition \ref{def_softaggr}.

\begin{figure}[t]
	\centering
	\includegraphics[trim=0.5cm 0.2cm 0 0,width=3.6in]{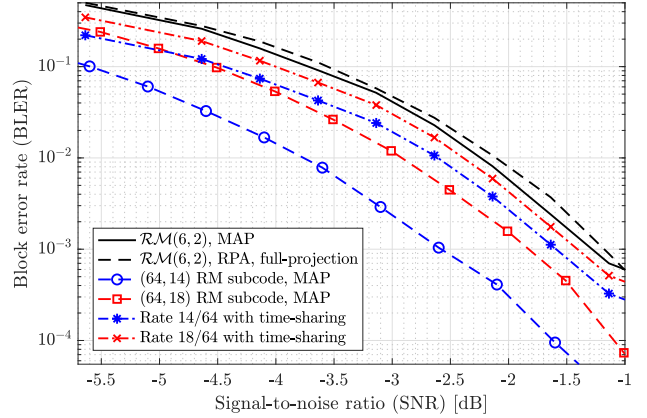}
	\caption{Simulation results for the BLER of various codes with MAP decoding. The comparison with the time-sharing scheme between $\mathcal{RM}(6,1)$ and $\mathcal{RM}(6,2)$ to achieve the same rates $14/64$ and $18/64$ is also included.}
	\label{fig1}
\end{figure}
\subsection{Encoding Insights}\label{sec_encoding}
The main objective of this paper is to develop schemes for decoding RM subcodes with low complexity. In this subsection, we provide some insights on how the design of the encoder can affect the decoding complexity as well as the performance. First, in order to further highlight the efficiency of RM subcodes, in Figure \ref{fig1}, we compare the block error rate (BLER) performance of RM subcodes with time-sharing (TS) between RM codes given the optimal MAP decoding.
We consider two RM subcodes with parameters $(n,k)=(64,14)$ and $(64,18)$. The generator matrix construction for these codes is based on having the largest ranks for the projected generator matrices which will be clarified at the end of this subsection. The TS performance is obtained assuming that the transmitter employs an $\mathcal{RM}(6,2)$ encoder in $\alpha$ portion of the time and an $\mathcal{RM}(6,1)$ encoder in the remaining $(1-\alpha)$ portion, where $\alpha=7/15$ and $11/15$, to achieve the same code rates $14/64$ and $18/64$, respectively. It is observed that the RM subcodes with the rates $14/64$ and $18/64$ achieve more than $1$ \si{dB} and $0.4$ \si{dB} gains, respectively, compared to the TS counterparts. Also, the performance of the RM subcode with rate $18/64$ is almost $0.2$ \si{dB} better than the performance of the lower rate code with TS. Note that all the simulation results in this paper are obtained from more than $10^5$ trials of random codewords (except $\mathcal{RM}(6,2)$ with MAP decoding that has $10^4$ trails). Additionally, throughout the paper, we define the signal-to-noise ratio (SNR) as ${\rm SNR}:=1/(2\sigma^2)$ and the energy-per-bit $E_b$ to the noise ratio as $E_b/N_0:=n/(2k\sigma^2)$, where $\sigma^2$ is the noise variance.

\begin{figure}[t]
	\centering
	\includegraphics[trim=0.5cm 0.2cm 0 0,width=3.6in]{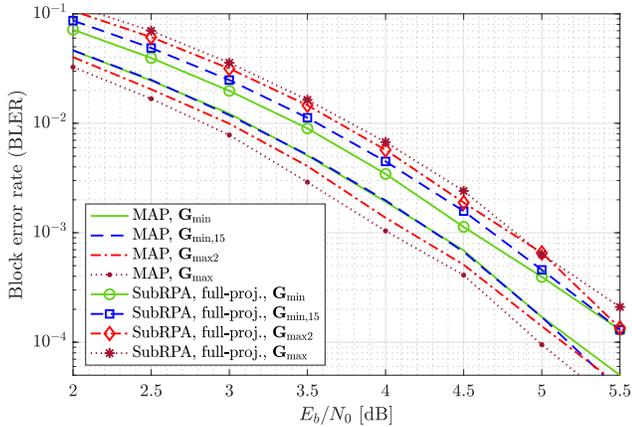}
	\caption{Simulation results for the $(64,14)$ RM subcodes with MAP and subRPA decoders given four different selections of the generator matrix $\boldsymbol{G}_{k\times n}$.}
	\label{fig2}
\end{figure}
As discussed earlier, our decoding algorithms perform MAP or soft-MAP decoding at the bottom layer. Also, the dimension of the projected codes at the bottom layer (i.e., the rank of the projected generator matrices) can be different. This is in contrast to RM codes that always result in the same dimension for the projected codes at the bottom layer. Therefore, an immediate approach for encoding RM subcodes to achieve a lower decoding complexity is to  construct the code generator matrix such that the projected codes at the bottom layer have smaller dimensions, and thus the decodings at the bottom layer have lower complexities. In other words, let $\mathcal{L}:=\sum_{t=1}^{T}2^{R_{t}}$ represent a rough evaluation of the decoding complexity at the bottom layer, i.e., the decoding complexity at the bottom layer is roughly a constant times $\cL$. Then, among all $\binom{k_u-k_l}{k-k_l}$ possible selections of the generator matrix $\boldsymbol{G}_{k\times n}$, we can choose the ones that achieve a smaller $\mathcal{L}$. This encoding scheme leads to reduction in the decoding complexity of our algorithms but it also affects the performance. In order to investigate the effect of this methodology, in Figure \ref{fig2}, we consider four different selections of the generator matrix for the $(64,14)$ RM subcode. In particular, $\boldsymbol{G}_{\rm max}$ and $\boldsymbol{G}_{{\rm max}2}$ have the largest and the second largest values of $\mathcal{L}=2568$ and $2532$, respectively. Also, $\boldsymbol{G}_{\rm min}$ has the minimum value of $\mathcal{L}=1482$. And, $\boldsymbol{G}_{{\rm min},15}$ has the minimum value of $\sum_{t}2^{R_{t}}=108$ on $15$ projections but a relatively large value of $\mathcal{L}=2412$ on all $63$ projections. Throughout the simulation results in this paper, the number of outer iterations for our recursive algorithms is set to $N_{\rm max}=3$ to assure the convergence of the algorithms. Figure \ref{fig2} suggests a slightly better performance for the MAP decoder for larger values of $\mathcal{L}$. However, surprisingly, our decoding algorithm exhibits a completely opposite behavior, i.e., a better performance is achieved for our subRPA algorithm with smaller values of $\mathcal{L}$. This is then a two-fold gain: a better performance for an encoding scheme that results in a lower complexity for our decoding algorithm. We did extensive sets of experiments which all confirm this \textit{empirical} observation. However, still, further investigation is needed to precisely characterize the performance-complexity trade-off as a result of the encoding process.

\subsection{Projection Pruning}\label{sec_projprun}
One direction for reducing the complexity of our decoding algorithms is to prune the number of projections at each layer. Particularly, let us assume that, at each layer and point in the decoding tree diagram, the complexity of decoding each branch (that corresponds to a given projection) is the same. This is not precisely true given that the projected codes at the bottom layer may have different dimensions. Also, we assume that the complexity of the aggregations performed at each layer is the same. Then, pruning the number of projections by a factor $\beta\in(0,1)$ is roughly equivalent to reducing the complexity by a factor of $\beta$ at each layer. In other words, if we have a subcode of $\mathcal{RM}(m,r)$, then there are $r-1$ layers in the decoding tree and hence, the projection pruning exponentially reduces the decoding complexity by a factor of $\beta^{r-1}$. This is essential to make the decoding of higher order RM subcodes practical. One can also opt to choose a constant number of projections per layer (i.e., prune the number of projections at upper layers with smaller factors) to avoid high-degree polynomial complexities.

Given that the projected codes at the bottom layer can have different dimensions (in contrast to RM codes), the projection subspaces should be carefully selected to reduce the complexity without having a notable effect on the decoding performance. Our empirical results show that the choice of the sets of projections can significantly affect the decoding performance.
To see that, in Figure \ref{fig3}, we consider the generator matrix $\boldsymbol{G}_{{\rm min},15}$ for encoding a $(64,14)$ RM subcode. In addition to full-projection decoding (i.e., $63$ 1-D subspaces), we also evaluate the performance of subRPA and soft-subRPA with $15$ projections picked according to three different projection pruning schemes. First, we consider a subset of $15$ subspaces that results in maximum ranks for the projected generator matrices at the bottom layer. In this setting, denoted by ``maxRank'' in Figure \ref{fig3}, all the $15$ projections result in the same rank of $6$. It is observed that this selection of the projections significantly degrades the performance (almost $1$ \si{dB} gap with full-projection decoding). Our extensive simulation results with other generator matrices and code parameters also confirm the same observation that, although it requires a higher complexity for MAP or soft-MAP decoding of the projected codes at the bottom layer, the ``maxRank'' selection fails to achieve a good performance compared to other considered pruning schemes.

Next, we consider the other extreme of projection selection, i.e., we select $15$ subspaces that result in minimum ranks for the projected codewords (``minRank'' scheme in Figure \ref{fig3}). In this case, three of the ranks are equal to $2$ and the remaining are equal to $3$. Therefore, the decoder in this case can perform the MAP and soft-MAP decodings at the bottom layer almost $8$ times faster than in maxRank selection. Surprisingly, despite its lower complexity compared to the maxRank selection, the minRank selection is capable of achieving very close to the performance of full-projection decoding ($\approx 0.1$ \si{dB} gap for both subRPA and soft-subRPA). Our additional simulation results also confirm the same observation and hence, establish the promising advantages of minRank projection pruning scheme in significantly reducing the decoding complexity while maintaining a negligible gap with the performance of full-projection decoding.

\begin{figure}[t]
	\centering
	\includegraphics[trim=0.5cm 0.2cm 0 0,width=3.6in]{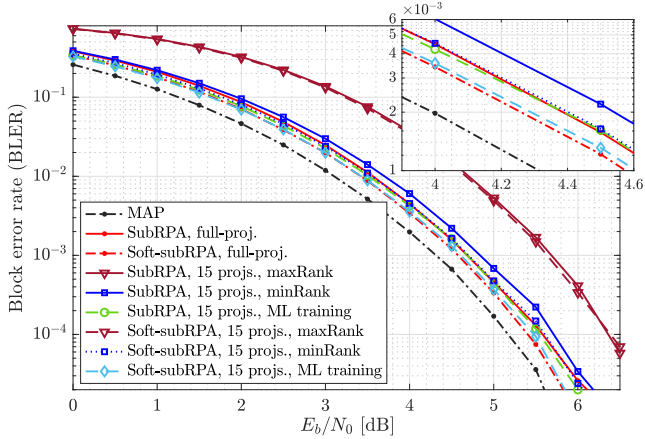}
	\caption{Performance of subRPA and soft-subRPA with full-projection decoding as well as different projection pruning schemes, i.e., picking according to the minimum ranks, maximum ranks, and training a machine learning model. The generator matrix $\boldsymbol{G}_{{\rm min},15}$ is considered for the encoding process.}
	\label{fig3}
\end{figure}

Even though minRank selection scheme is capable of achieving very close to the performance of full-projection decoding, we cannot guarantee that it is the best selection in terms of minimizing the decoding error rate. In practice, we may want to prune most of projections per layer to allow efficient decoding at higher rates (equivalently, higher order RM subcodes) with a manageable complexity. In such scenarios, we may, inevitably, have a meaningful gap with full-projection decoding, more than what we observed here for minRank selection (i.e., $\approx 0.1$ \si{dB}). Therefore, one needs to ensure that the sets of the selected projections are the ones that minimize the decoding error rate, i.e., the gap to the full-projection decoding. In the next subsection, we shed light on how the proposed soft-subRPA algorithm enables training a ML model to search for optimal sets of projections. This then establishes that the combination of our soft-subRPA with our ML model enables efficient decoding (in terms of both decoding error rate and complexity) of RM subcodes. To see the potentials of this scheme, in Figure \ref{fig3} we included the results of our decoding algorithms with $15$ projections obtained via training our ML model. It seems that the trained model also has tendency to pick projections that result in smaller ranks for the projected generator matrices, i.e., $3$ rank-$2$, $6$ rank-$3$, and $6$ rank-$4$ projections are picked by the ML model.
Figure \ref{fig3} demonstrates identical performance to full-projection decoding, for both subRPA and soft-subRPA algorithms, which is the best one can hope for with the pruned-projection decoding. Additionally, it is observed that the soft-subRPA algorithm can almost $0.1$ \si{dB} improve upon the performance of the subRPA algorithm.

\subsection{Training an ML Model for Projection Pruning}
\label{sec_nn}
As explained earlier, the goal is to train an ML model to find the best subset of projections. To do so, we assign a weight metric $w_q$ to each $q$-th projection such that $w_q\in[0,1]$ and $\sum_{q=1}^Qw_q=1$, where $Q$ is the number of full projections for a given (projected) code in the decoding process. The objective is then to train an ML model  to pick a subset of $Q_0$ projections (i.e., prune the number of projections by a factor $\beta=Q_0/Q$) that minimize the training loss. Building upon the success of stochastic gradient descent methods in training complex models, we want to use gradients for this search. In other words, the ML model updates the weight vector $\boldsymbol{w}:=(w_q, q\in[Q])$ such that picking the $Q_0$ projections corresponding to the largest weights results in the best performance.

There are two major challenges in training the aforementioned ML model. First, the MAP decoding that needs to be performed at the bottom layer (see \eqref{map}) is not differentiable since it involves the ${\rm argmax}(\cdot)$ operation which is not a continuous function. Therefore, one cannot apply the gradient-based training methods to our subRPA algorithm. However, the proposed soft-subRPA algorithm overcomes this issue by replacing the non-differentiable MAP decoder at the bottom layer with the differentiable soft-MAP decoder\footnote{Note that the soft-MAP algorithm involves $\max(\cdot)$ function which, unlike ${\rm argmax}(\cdot)$, is a continuous function. Also, the derivative of the function $\max(0,x)$ is defined everywhere except in $x=0$ which is a rare event to happen. Accordingly, advanced training tools, such as PyTorch library (that is used in this research), easily handle and treat $\max(\cdot)$ as a differentiable function. For example, the rectified linear unit function ${\rm ReLU}(x):=\max(0,x)$ is a widely used activation function in  neural networks.}. 
The second issue is that the combinatorial selection of $Q_0$ largest elements of the vector $\boldsymbol{w}$ is not differentiable. To address this issue, we apply the SOFT (Scalable Optimal transport-based diFferenTiable) top-$k$ operator,
proposed very recently in \cite{xie2020differentiable}, to obtain a smoothed approximation of the top-$k$ operator whose gradients can be efficiently approximated. It is worth mentioning that the SOFT top-$k$ function is a generalization of the soft-max function, which is a soft version of the  $\rm{argmax}$ function. In other words, the SOFT top-$k$ function can be viewed as a soft version of top-$k$ function.

Next, the training procedure is briefly explained. We use the PyTorch library of Python to first implement our soft-subRPA decoding algorithm in a fully differentiable way for the purpose of gradient-based training. We initialize the weight vector as $\boldsymbol{w}_0:=(1/Q,\cdots, 1/Q)$, i.e., equal weights for all the projections. For each training iteration, we randomly generate a batch of $B$ codewords of the RM subcode, and compute their corresponding LLR vectors given a carefully chosen training SNR. Then we input these LLR vectors to our decoder to obtain the soft decisions at each layer. During the soft-aggregation step, instead of unweighted averaging of \eqref{eq_softaggr}, we take the weighted averages of the soft decisions at all $Q$ projections as $\tilde{\boldsymbol{l}}(\boldsymbol{z})=\sum_{q=1}^{Q}w_q\tanh\big(\boldsymbol{\hat{l}}_q\left([\boldsymbol{z}+\mathbb{B}_q]\right)/2\big) \boldsymbol{l}(\boldsymbol{z}\oplus\boldsymbol{z}_q)$.
Ideally, the top-$k$ operator should return nonzero weights only for the top $Q_0$ elements. However, due to the smoothed SOFT top-$k$ operator, all $Q$ elements of $\boldsymbol{w}$ may get nonzero weights though the weights for the $Q-Q_0$ smaller elements are very small. Therefore, the above weighted average is approximately equal to the weighted average over only the largest $Q_0$ weights (i.e., pruned-projection decoding).
Note that we apply the same procedure for all (projected) RM subcodes at each node and layer of the recursive decoding algorithm while we define different weight vectors (and also $Q_0$'s) for each sets of projections corresponding to each (projected) codes. We also consider fixed weight vectors for decoding all $B$ codewords at each iteration.

The ML model then updates all weight vectors at each iteration to iteratively minimize the training loss. To do so, we apply the ``Adam'' optimization algorithm \cite{kingma2014adam} to minimize the training loss while using ``BCEWithLogitsLoss'' \cite{nn_loss} as the loss function which efficiently combines a sigmoid layer with the binary cross-entropy (BCE) loss. By computing the loss function between the true labels from the generated codewords and the predicted LLRs from the decoder output, the optimizer then moves one step forward by updating the model, i.e., the weight vectors. Finally, once the model converges after enough number of iterations, we save the weight vectors for the sake of optimal projection pruning. Note that in order to reduce the decoding complexity and the overload of training process, we only train the model for a given, properly chosen, training SNR. In other words, once the training is completed, we fix the subsets of projections according to the largest values of the weight vectors. We then test the performance of our algorithms given the fixed decoder (i.e., the fixed subsets of projections) for all codewords and all SNR points. One can apply the same procedure to train the model for each SNR point, or even actively for each LLR vector, to possibly improve upon the performance of our \textit{fixed} projection pruning scheme at the expense of increased training overload.

\section{Conclusions}\label{conc}
In this paper, we designed efficient decoding algorithms for decoding subcodes of RM codes. 
More specifically,
we first proposed a general recursive algorithm, called subRPA, for decoding RM subcodes. Then we derived a soft-decision based version of our algorithm, called soft-subRPA, that not only improved upon the performance of the subRPA algorithm but also enabled a differentiable implementation of our decoding algorithm for the purpose of training a machine learning model. Accordingly, 
we proposed an efficient pruning scheme that finds the best subsets of projections via training a machine learning model.
Our simulation results on a $(64,14)$ RM subcode demonstrate as good as the performance of full-projection decoding for our machine learning-aided decoding algorithms with more than $4$ times smaller number of projections.
The research in this paper can be extended in many directions such as training machine learning models to design efficient encoders for RM subcodes and also leveraging higher dimension subspaces for projections to, possibly, further reduce the decoding complexity.

\appendices
\section{LLRs of the Information Bits}\label{appnd_LLR_inf}
Consider an AWGN channel model as $\boldsymbol{y}=\boldsymbol{s}+\boldsymbol{n}$, where $\boldsymbol{s}=1-2{\boldsymbol{c}}$, $\boldsymbol{c}\in \mathcal{C}$, and $\boldsymbol{n}$ is the AWGN vector with mean zero and variance $\sigma^2$ elements. Then, the LLR of the $i$-th information bit $u_i$ can be obtained using max-log approximation as
\begin{align}\label{llr_inf}
\boldsymbol{l}_{\rm inf}(i) \approx \operatorname*{max}_{\boldsymbol{c}\in\mathcal{C}_i^0}~\langle \boldsymbol{l}, 1-2{\boldsymbol{c}}\rangle~-~ \operatorname*{max}_{\boldsymbol{c}\in\mathcal{C}_i^1}~\langle \boldsymbol{l}, 1-2{\boldsymbol{c}}\rangle,
\end{align}
where $\boldsymbol{l}:=2\boldsymbol{y}/\sigma^2$ is the LLR vector of the AWGN channel, and $\mathcal{C}_i^0$ and $\mathcal{C}_i^1$ are subsets of codewords that have the $i$-th information bit $u_i$ equal to zero or one, respectively.
To see this, observe that
\begin{align}\label{llrinf1}
\boldsymbol{l}_{\rm inf}(i) :=& \ln\left(\frac{\Pr(u_i=0|\boldsymbol{y})}{\Pr(u_i=1|\boldsymbol{y})}\right)\nonumber\\
\stackrel{(a)}{=}&\ln\left(\frac{\sum_{\boldsymbol{s}\in\mathcal{C}_i^0}\exp\left(-||\boldsymbol{y}-\boldsymbol{s}||_2^2/\sigma^2\right)}{\sum_{\boldsymbol{s}\in\mathcal{C}_i^1}\exp\left(-||\boldsymbol{y}-\boldsymbol{s}||_2^2/\sigma^2\right)}\right)\nonumber\\
\stackrel{(b)}{\approx}&\frac{1}{\sigma^2}\operatorname*{min}_{\boldsymbol{c}\in\mathcal{C}_i^1}||\boldsymbol{y}-\boldsymbol{s}||_2^2-\frac{1}{\sigma^2} \operatorname*{min}_{\boldsymbol{c}\in\mathcal{C}_i^0}||\boldsymbol{y}-\boldsymbol{s}||_2^2,
\end{align}
where step $(a)$ is by applying the Bayes' rule, the assumption $\Pr(u_i=0)=\Pr(u_i=1)$, the law of total probability, and the distribution of Gaussian noise. Moreover, step $(b)$ is by the max-log approximation. Finally, given that all $\boldsymbol{s}$'s have the same norm, we obtain \eqref{llr_inf}.



\end{document}